\documentclass{article}

\usepackage{authblk}
\usepackage{amsthm}
\usepackage[ruled]{algorithm2e} 

\newtheorem{theorem}{Theorem}
\newtheorem*{theorem*}{Theorem}
\newtheorem{lemma}[theorem]{Lemma}
\theoremstyle{definition}
\newtheorem{definition}[theorem]{Definition}

\newcommand{\argmax}{\mbox{argmax}}

\begin{document}

\title{Online Matching in a Ride-Sharing Platform}  
\author{Chinmoy Dutta \thanks{Lyft, San Francisco. Email: cdutta@lyft.com} \qquad Chris Sholley \thanks{Lyft, San Francisco. Email: chris@lyft.com}}
\date{} 

\maketitle

\begin{abstract}
We propose a formal graph-theoretic model for studying the problem of matching rides online in a ride-sharing platform. Unlike most of the literature on online matching, our model, that we call {\em Online Windowed Non-Bipartite Matching} ($\mbox{OWNBM}$), pertains to online matching in {\em non-bipartite} graphs. We show that the edge-weighted and vertex-weighted versions of our model arise naturally in ride-sharing platforms. We provide a randomized $\frac{1}{4}$-competitive algorithm for the edge-weighted case using a beautiful result of Lehmann, Lehmann and Nisan (EC 2001) for combinatorial auctions. We also provide an $\frac{1}{2} (1 - \frac{1}{e})$-competitive algorithm for the vertex-weighted case (with some constraint relaxation) using insights from an elegant randomized primal-dual analysis technique of Devanur, Jain and Kleinberg (SODA 2013). 
\end{abstract}

\section{Introduction}
\label{sec:intro}
In recent times, the rapid advancement and widespread availability of mobile computing have significantly uplifted the face of urban mobility and transportation. A large part of this change is due to the impact created by ride-sharing platforms. Such platforms allow any rider to request a ride at any time and any place at the convenience of a tap on a ride-sharing app installed on her mobile device. On receiving a ride request, the platform fulfils it by choosing one of the nearby drivers available to transport the rider and dispatching her to pick up the rider. Besides the convenience of the ride-hailing process, these platforms can also do marketplace optimization to ensure lower estimated time to pick-up for riders as well as better time utilization for drivers. 

The major advantage of these platforms, however, is the {\em ride-sharing} possibility that they enable. With this feature, riders can choose to share the vehicle with other riders heading in the same direction. There are tremendous socio-economical and environmental advantages to this sharing. It lets the riders share the cost of the ride, thereby providing them an economical and affordable means of transportation. Riders can also take advantage of incentives like use of high occupancy lanes. Moreover, ride-sharing eases the burden on the transportation infrastructure of cities and helps reduce traffic congestion, specially during heavy commute times. Reducing traffic congestion in turn helps in cutting down commute times, thus providing great economic benefit to individuals, businesses and institutions by reducing lost time and increasing productivity. Perhaps even more importantly, ride-sharing provides great environmental benefits by reducing air pollution as a result of reduced number of cars on roads, thereby greatly reducing our carbon footprint. In fact, for some densely populated and highly congested cities, intelligent coupling of ride-sharing with smart transit systems seems the only viable and scalable transportation option.

Sharing rides with others, however, does incur some costs to the riders like increased detours and travel times as well as some loss of privacy and convenience. The task of the ride-sharing platform is to intelligently trade off these costs against the benefits mentioned above. To be able to match rides that incur minimal inconvenience to the riders while maximizing marketplace efficiency, therefore, is of paramount importance for a ride-sharing platform. Considering the central role of ride-sharing platforms for transportation and economy, and the central role of this matching problem for such platforms, we strongly believe that there is need for modeling this problem formally and studying it mathematically. To the best of our knowledge, such rigorous attempt has not been made for this problem yet.

In this paper, we identify the essential aspects of the problem of matching rides in a ride-sharing platform. This then enables us to formally model the problem in graph-theoretic online matching setting. Our modeling leads to new problems in the space of online matching that we believe are of independent theoretical interest and importance. We then provide approximation algorithms for two different variants of our proposed model that arise naturally in the ride-sharing context.

\subsection{Matching Rides for Ride-sharing}
\label{subsec:ride-share}
We start by identifying the unique attributes of the problem of matching rides for ride-sharing. We discuss both the constraints that must be satisfied and the objectives that might be optimized subject to satisfying those constraints. 

\subsubsection{Constraints}
\label{subsubsec:constraints}

\paragraph{Spatial constraint} Matching rides obviously has a spatial constraint to it. That is, rides should be matched only if there is good spatial overlap among their routes. While the exact notion and criteria of spatial overlap is for the platform to decide, this ensures that none of the riders will experience an unacceptable amount of detour or increase in their travel time. With this spatial criteria in place, the existing rides in the platform that a new ride can be matched to is known once the new ride arrives. 

\paragraph{Temporal constraint} Matching rides also has a temporal constraint to it. For example, two rides cannot be matched together even if they have perfect spatial overlap (that is, exact same pick-up and drop-off locations), if they arrive at very different times. The acceptable maximum period between the arrival of two rides such that they can still be matched together is for the ride-sharing platform to determine. We call this quantity the {\em incidence window}.

\paragraph{Online constraint} Furthermore, the ridesharing platform cannot wait too long for making its decisions and must make them in real-time as rides arrive. In other words, the problem must be solved {\em online} without full knowledge of rides to arrive in future. The decision that the platform needs to make for every arriving ride is whether to {\em match-or-dispatch}. If the platform decides to match the ride, it must pick an existing ride to match it to. Otherwise, the platform must dispatch a driver to pick-up the ride. The maximum period the platform can wait since the arrival of a ride before it must make a match-or-dispatch decision for it is called the {\em matching window}.

\subsubsection{Objectives}
\label{subsubsec:objectives}

\paragraph{Matching efficiency objective} A widely used objective to optimize for while matching rides in a ride-sharing platform is the matching efficiency. A feasible match that satisfies the constraints above provides some efficiency for the system. For example, one notion of efficiency can be the amount of savings we get by making the match as compared to transporting the riders individually. The savings may be measured in terms of distance driven, time taken, cost of trips or any combination of these. Once the platform fixes the notion of efficiency, the goal it can optimize for could be the system-wide matching efficiency which is the sum of the efficiencies derived from all the matches made.

\paragraph{Matching rate objective} Another objective that the platform might want to optimize for while matching rides is the matching rate. This is especially important in the growth stages of a ride-sharing platform where it wants to maximize the volume of shared rides subject to meeting the matching constraints. Different rides can have different value for the platform depending upon various factors like distance of the route, origin-destination locations, profitability etc. In this case, the goal the platform can optimize for is the matching rate which is the sum of the values of all the rides that were matched.

Both matching efficiency and matching rate are widely used and very closely watched metrics in the ride-sharing industry.

\subsection{Our contributions}
\label{subsec:contrib}
In this paper, we provide a formal model for the problem of matching rides for ride-sharing in a graph-theoretic online matching setting. We call our model {\em online windowed non-bipartite matching} ($\mbox{OWNBM}$). In this model, time is considered discrete and vertices arrive sequentially, one per time step. The edges in the graph represent matchability (that is, satisfying the spatial and temporal constraints). The edges originating from a vertex are revealed once the vertex arrives and each such edge can go to one of the preceding $d$ (incidence window) vertices. It helps to think of the vertices as ordered by arrival times and edges as directed and going from a vertex to a preceding one by that ordering. After $w$ (matching window) time steps since the arrival of a vertex, we must match it to one of the (yet unmatched) vertices it has an edge to or from, or leave it unmatched. A matching is a subset of edges such that every vertex has at most one edge incident on it (either to or from) that is included in the subset.

We consider two variants of our model. In the {\em edge-weighted} variant, each edge of the graph has a weight and the objective is to construct a matching to maximize the sum of the weights of the edges picked in the matching. This corresponds to the matching efficiency objective for the ride-sharing platform mentioned above. In the {\em vertex-weighted} variant of our model, each vertex of the graph has a weight and the objective is to construct a matching to maximize the sum of the weights of the vertices that get matched (has an edge of the matching incident on it, either to or from). This corresponds to the matching rate objective mentioned above.

We provide the following approximation algorithms for the edge and vertex weighted versions of windowed online non-bipartite matching.

\begin{theorem}
\label{thm:edge-weighted}
There is a randomized $\frac{1}{4}$-approximation for edge-weighted $\mbox{OWNBM}$ with $w = d$, where $w$ and $d$ are the matching and incidence windows respectively.
\end{theorem}

\begin{theorem}
\label{thm:vertex-weighted}
There is a randomized  $\frac{1}{2} (1 - \frac{1}{e})$-approximation for vertex-weighted $\mbox{OWNBM}$ with $w = d$, where $w$ and $d$ are the matching and incidence windows respectively. This is under the assumption that the online algorithm is allowed to match three rides together and is allowed to delete an edge from a $3$-match.
\end{theorem}

Note that in Theorem~\ref{thm:vertex-weighted} above, while the online solution is allowed to match three rides together, the competitive ratio is measured against the offline optimal which only matches two rides together. (See details in sec~\ref{sec:model}.)

To the best of our knowledge, online matching has not yet been studied for non-bipartite graphs in adversarial settings. Our model is simple and yet powerful to capture the essential aspects of matching rides in a ride-sharing platform. We believe one of the main novelties of our work is that our model allows for studying online matching in non-bipartite graphs. The notion of incidence and matching windows that we propose leads to an interesting theoretical model that we believe is of independent interest and can find further use outside of the immediate context of ride-sharing.

\subsection{Related work}
\label{subsec:related}
Recently, there has been a lot of research interest in ride-sharing and associated challenges and problems. Much of the work has been done in transportation-related literature. A recent study by Santi et al~\cite{S14} showed about $80\%$ of rides in Manhattan can be shared by two riders. A lot of studies related to fleet management considers ride-sharing without pooling requests, e.g~\cite{PSFR12, ZP14, S14, SSGF16}. There has also been a lot of research interest in studying autonomous ride-sharing systems~\cite{PSFR12, SSGF16, CA16}. Heuristic-based solutions to the rider and driver assignment problem were studied in~\cite{AESW11, MZW13}. Alonso-Mora et al.~\cite{ASWFR17} studied real-time high-capacity ride-sharing and route generation. However, their work does not provide a theoretical graph-theoretic online matching formulation for matching rides.

Starting with the seminal work of Karp, Vazirani and Vazirani~\cite{KVV90}, online matching has been widely studied and a vast amount of literature exists on the topic. However, such literature, being mostly motivated by the real-life applications such as online advertising, kidney exchange and online dating, considers the problem in the bipartite setting. The book by Mehta~\cite{M12} gives a detailed overview of this literature. The work of Karp, Vazirani and Vazirani~\cite{KVV90} gave an optimal online algorithm with competitive ratio $(1 - \frac{1}{e})$ for unweighted bipartite matching in the adversarial arrival model. The vertex-weighted version was studied by Aggarwal, Goel, Karande and Mehta~\cite{AGKM11} where they gave an optimal $(1 - \frac{1}{e})$ ratio. The edge-weighted version was studied under the additional economic assumption of free disposal by Feldman, Korula, Mirrokni, Muthukrishnan and Pál~\cite{FKMM09} where they gave a $(1 - \frac{1}{e})$ competitive ratio.

Feldman, Mehta, Mirrokni and Muthukrishnan~\cite{FMMM09} were the first to beat the $(1 - \frac{1}{e})$ bound for the unweighted version assuming IID arrival model and gave a competitive ratio of $0.67$. This was improved by Manshadi, Gharan and Saberi~\cite{MGS12} to $0.705$. Jaliet and Lu~\cite{JL13} presented adaptive algorithms to further improve the ratio to $0.729$ for the unweighted variant, and also achieved a ratio of $0.725$ for the vertex-weighted variant. Haeupler, Mirrokni and Zadimoghaddam~\cite{HMZ11} first beat the $(1 - \frac{1}{e})$ bound for the edge-weighted variant and achieved a ratio of $0.667$.

Online matching literature for the  non-bipartite setting is extremely limited. Anderson, Ashlagi, Gamarnik and kanoria~\cite{AAGK15} studied a dynamic model of matching under homogeneous and independent stochastic assumptions to minimize average wait-times for agents. Akbarpour, Li and Gharan~\cite{ALG14} studied dynamic matching with stochastically arriving and departing agents. However, as stated earlier, there has not been any work on online matching in the adversarial setting to the best of our knowledge.

\subsection{Organization of the paper}
\label{subsec:organize}
We formally present our model, definitions and results in Section~\ref{sec:model}. The proof of our result for edge-weighted $\mbox{OWNBM}$ (Theorem~\ref{thm:edge-weighted}) is presented in Section~\ref{sec:edge-weighted}, and that for the vertex-weighted case (Theorem~\ref{thm:vertex-weighted}) is presented in Section~\ref{sec:vertex-weighted}. Finally, we conclude in Section~\ref{sec:conclude} with discussion of open problems and future directions for research.

\section{Our Model}
\label{sec:model}
In this section, we formalize definitions and notations related to our model. We also restate our results formally.

\begin{definition}[Online Windowed Non-Bipartite Matching ($\mbox{OWNBM}$)]
An instance $\mathcal{I}$ of the Online Windowed Non-Bipartite Matching problem is a tuple $(\mathcal{G}, d)$ where $\mathcal{G}$ is a graph and $d$ is a non-negative integer. It satisfies the following properties.
\begin{itemize}
\item Vertices of $\mathcal{G}$ arrive sequentially, one at each time step.
\item (Directed) edges of $\mathcal{G}$ originating from a vertex are revealed when the vertex arrives. 
\item {\bf [incidence window ($d$)]} Every (directed) edge in $\mathcal{G}$ originating from a vertex terminates at one of the $d$ preceding vertices. 
\end{itemize}
A matching $\mathcal{M}$ in $\mathcal{I}$ is a subset of edges of $\mathcal{G}$ such that for every vertex $v$ of $\mathcal{G}$, at most one edge of $\mathcal{M}$ is incident on $v$ (either originating from it or terminating at it).
\end{definition}

\begin{definition}[Edge-weighted $\mbox{OWNBM}$]
An edge-weighted $\mbox{OWNBM}$ is an instance $\mathcal{I} = (\mathcal{G}, d)$ of $\mbox{OWNBM}$, where every edge $e$ of $\mathcal{G}$ has an associated non-negative weight $w_e$. The weight of a matching $\mathcal{M}$ in $\mathcal{I}$ is $\sum_{e \in \mathcal{M}} w_e$.
\end{definition}

\begin{definition}[Vertex-weighted $\mbox{OWNBM}$]
A vertex-weighted $\mbox{OWNBM}$ is an instance $\mathcal{I} = (\mathcal{G}, d)$ of $\mbox{OWNBM}$, where every vertex $v$ of $\mathcal{G}$ has an associated non-negative weight $w_v$. The weight of a matching $\mathcal{M}$ in $\mathcal{I}$ is $\sum_{(u,v) \in \mathcal{M}} w_u + w_v$.
\end{definition}

\begin{definition}[$(r, w)$-Competitive algorithm for $\mbox{OWNBM}$]
An ($r$, $w$)-Competitive algorithm $\mbox{ALG}$ for $\mbox{OWNBM}$ with $r \leq 1$ and $w \geq 0$ is an algorithm that, for every instance $\mathcal{I} = (\mathcal{G}, d)$ of $\mbox{OWNBM}$, constructs a matching $\mathcal{M}$ in $\mathcal{I}$ online such that 
\begin{itemize}
\item {\bf [matching window ($w$)]} any edge $e = (j, i)$ terminating at a vertex $i$ included in $\mathcal{M}$ is picked no later than $w$ time steps after the arrival of $i$ (and hence also $j$). An edge once picked cannot be deleted from $\mathcal{M}$ later.
\item {\bf [competitive ratio ($r$)]} $\frac{\mbox{ALG}_\mathcal{I}}{\mbox{OPT}_\mathcal{I}} \geq r$, where $\mbox{ALG}_\mathcal{I}$ is the weight of the matching $\mathcal{M}$ and $\mbox{OPT}_\mathcal{I}$ is the maximum (optimal) weight of a matching in $\mathcal{I}$. The weight of the matching is defined depending on whether it is an edge-weighted or a vertex-weighted instance of $\mbox{OWNBM}$. In case $\mbox{ALG}$ is randomized, $\mathcal{M}$ is a random variable and $\mbox{ALG}_\mathcal{I}$ is the expected weight of $\mathcal{M}$ over the randomness used in $\mbox{ALG}$.
\end{itemize}
\end{definition}

We now restate our main results that we already presented earlier.

\begin{theorem*}[Restatement of Theorem~\ref{thm:edge-weighted}]
There is a randomized $(\frac{1}{4}, d)$-competitive algorithm for edge-weighted $\mbox{OWNBM}$ where $d$ is the incidence window of the $\mbox{OWNBM}$ instance. 
\end{theorem*}

\begin{theorem*}[Restatement of Theorem~\ref{thm:vertex-weighted}]
There is a randomized $(\frac{1}{2} (1 - \frac{1}{e}), d)$-competitive algorithm for vertex-weighted $\mbox{OWNBM}$ where $d$ is the incidence window of the $\mbox{OWNBM}$ instance, if we allow the online algorithm (but not the offline optimal algorithm) to 
\begin{enumerate}
\item form a $3$-matching,and 
\item delete an edge $(j, i)$ from a $3$-set no later than $d$ steps since the arrival of vertex $j$. 
\end{enumerate}
A 3-matching $\mathcal{M}$ in $\mathcal{I}$ is a set of disjoint subsets of vertices of $\mathcal{G}$, each of size $2$ or $3$, such that each $2$-set is an edge in $\mathcal{G}$, and the induced subgraph of each $3$-set has at least $2$ edges. The weight of a $3$-matching is defined as the sum of weight of all vertices included a $2$-set or a $3$-set.
\end{theorem*}

Note that we do not assume any distribution on edges for our results above. In particular, our results hold for the adversarial model for the arriving vertices and edges.

It will be convenient to define a notion of {\em semi-matching} for the proofs of both the above theorems.

\begin{definition}[Semi-matching]
Given an instance $\mathcal{I} = (\mathcal{G}, d)$ of $\mbox{OWNBM}$, a semi-matching $\mathcal{M'}$ in $\mathcal{I}$ is a subset of edges of $\mathcal{G}$ such that for every vertex $v$ of $\mathcal{G}$, at most one edge of $\mathcal{M'}$ originates from $v$ {\em and} at most one edge of $\mathcal{M'}$ terminates at $v$.  
\end{definition}

\section{Edge-Weighted $\mbox{OWNBM}$}
\label{sec:edge-weighted}
In this section, we prove Theorem~\ref{thm:edge-weighted}. We will need a beautiful result from combinatorial auction theory by Lehmann, Lehmann and Nisan~\cite{LLN01}. First, we define an instance of a combinatorial auction from a given instance of $\mbox{OWNBM}$.

\begin{definition}
\label{def:auction}
Let $\mathcal{I} = (\mathcal{G}, d)$ be an instance of edge-weighted $\mbox{OWNBM}$, and let $\mathcal{G} = (V, E)$ where $V$ and $E$ are the sets of vertices and edges of $\mathcal{G}$ respectively. Let $|V| = n$. Define a combinatorial auction with $n$ bidders and $n$ items. Without loss of generality and slightly abusing notation, we assume the vertex set, the item set and the bidder set to be $[n]$. The valuation $\mbox{val}_i$ of the $i$'th bidder for a set $S$ of items is given by the following.
\[
\mbox{val}_i(S) = \max \{\max_{j \in S;(j,i) \in E} w_{(j,i)}, 0\}
\]
\end{definition}

The following lemma upper bounds the weight of optimal matching in a edge-weighted $\mbox{OWNBM}$ by the valuation of optimal allocation of the associated combinatorial auction.

\begin{lemma}
\label{lem:auction-ub}
Let $\mathcal{I} = (\mathcal{G}, d)$ be an instance of edge-weighted $\mbox{OWNBM}$ and $\mbox{OPT}_\mathcal{I}$ be the weight of an optimal matching. Then the optimal allocation for the combinatorial auction as defined in Definition~\ref{def:auction} has valuation at least $\mbox{OPT}_\mathcal{I}$.
\end{lemma}

\begin{proof}
To prove the lemma, we produce an allocation with valuation $\mbox{OPT}_\mathcal{I}$. Consider an optimal matching $\mathcal{M}$ in $\mathcal{I}$. For any edge $(j, i) \in \mathcal{M}$, set $S_i \leftarrow \{i, j\}$ and $S_j \leftarrow \emptyset$. For any unmatched vertex $k$ of $\mathcal{G}$, set $S_k \leftarrow \{k\}$. It is easy to see that the valuation of this allocation is $\mbox{OPT}_\mathcal{I}$, which proves the lemmma.
\end{proof}

Next, we show that the valuation functions for the auction derived from an edge-weighted $\mbox{OWNBM}$ are sub-modular.

\begin{lemma}
\label{lem:submodular}
The valuation functions for all the bidders in Definition~\ref{def:auction} are sub-modular.
\end{lemma}

\begin{proof}
Fix bidder $i$ and consider its valuation function $\mbox{val}_i$ as defined in Definition~\ref{def:auction}. Also consider sets $S \subseteq T$ of items and an item $k$. We will show that the marginal valuation of item $k$ given $T$ is at most the marginal valuation of $k$ given $S$, i.e., $\mbox{val}_i(k | T) \leq \mbox{val}_i(k | S)$. Assume the contrary. It is not hard to see that this will imply that
\[
\max_{j \in S;(j,i) \in E} w_{(j,i)} > \max_{j \in T;(j,i) \in E} w_{(j,i)}.
\]
Since $S \subseteq T$, this is an impossibility.
\end{proof}

Lemma~\ref{lem:submodular} lets us use the following 2-approximation from ~\cite{LLN01} for combinatorial auctions with sub-modular valuations. (Holds for $n$ bidders and $m$ items; in our case, $n = m$.)

\begin{algorithm}[t]
	\SetAlgoNoLine
	\KwIn{$n$ bidders with submodular valuations and $m$ items}
	\KwOut{A $\frac{1}{2}$-approximation to the optimal allocation}
	$S_1 = S_2 = \ldots = S_n \leftarrow \emptyset$\;
	\For{each item $j = 1, 2, \ldots, m$}
	{
		let $i$ be the bidder with highest marginal valuation $\mbox{val}_i(j | S_i)$ for item $j$\;
		$S_i \leftarrow S_i \cup \{j\}$\;
	}
	return $S_1 = S_2 = \ldots = S_n$;
	\caption{$\frac{1}{2}$-approximation for combinatorial auctions with submodular valuations}
	\label{alg:auction}
\end{algorithm}

\begin{theorem*}[Theorem 11 of~\cite{LLN01}]
\label{thm:auction-approx}
If the valuation functions of all the bidders are submodular, then Algorithm~\ref{alg:auction} is a $\frac{1}{2}$-approximation for the optimal allocation.
\end{theorem*}

A useful property of the above theorem is that it works in our online setting as shown in the following lemma.

\begin{lemma}
\label{lem:auction-approx}
Let $\mathcal{I} = (\mathcal{G}, d)$ be an instance of edge-weighted $\mbox{OWNBM}$ and $\mbox{OPT}_\mathcal{I}$ be the weight of an optimal matching. Then we can construct an allocation for the combinatorial auction as defined in Definition~\ref{def:auction} online with valuation at least $\frac{1}{2} \mbox{OPT}_\mathcal{I}$. Here, the $i$'th bidder and the $i$'th item arrives with the arrival of vertex $i$.
\end{lemma}

\begin{proof}
From Theorem~\ref{thm:auction-approx} and Lemma~\ref{lem:submodular}, we immediately see that we can construct an allocation with valuation at least half that of the optimal allocation offline once all the vertices of $\mathcal{G}$ have arrived. Combining with Lemma~\ref{lem:auction-ub} gives the valuation of this offline allocation to be at least $\frac{1}{2} \mbox{OPT}_\mathcal{I}$. The lemma follows by noting that the allocation can be done online (that is, items and bidders corresponding to a vertex $i$ arrives with the arrival of vertex $i$ and item $i$ is allocated to a bidder as soon as the item arrives) due to the following two facts.
\begin{enumerate}
\item the $\frac{1}{2}$-approximation algorithm presented above can be run sequentially as items arrive.
\item the marginal valuation of an item for a bidder that has not arrived yet when the item arrives is $0$.
\end{enumerate}
\end{proof}

The following lemma shows that we can do an approximation-preserving transformation from a valuation to a semi-matching online.

\begin{lemma}
\label{lem:edge-semi-matching}
Let $\mathcal{I} = (\mathcal{G}, d)$ be an instance of edge-weighted $\mbox{OWNBM}$ and $\mbox{OPT}_\mathcal{I}$ be the weight of an optimal matching. Then we can construct a
semi-matching $\mathcal{M'}$ of weight $\frac{1}{2} \mbox{OPT}_\mathcal{I}$ online such that any edge $(j, i) \in \mathcal{M'}$ terminating at vertex $i$ is picked no later than $d$ time steps after its arrival. 
\end{lemma}

\begin{proof}
Start by setting $\mathcal{M'} \leftarrow \emptyset$. Construct the allocation for the combinatorial auction online as in the proof of Lemma~\ref{lem:auction-approx}. Note that an item $j$ is allocated as soon as it arrives. As a result, given the incidence window property of $\mathcal{I}$, the allocation of bidder $i$ , $S_i$ is frozen $d$ time steps after arrival of vertex $i$. For a vertex $i$, if its valuation $\mbox{val}(S_i) = 0$ after $d$ steps from its arrival, then do not pick any edge terminating at vertex $i$. Otherwise, set $\mathcal{M'} \leftarrow \mathcal{M'} \cup \{(j, i)\}$, where
\[
j = \argmax_{k \in S_i} w_{(k, i)}.
\]
It is clear that $\mathcal{M'}$ is a semi-matching as every item is allocated to at most one bidder which means at most one edge of $\mathcal{M'}$ originates from a given vertex $j$. Also, by construction, we include at most one edge terminating at a given vertex $i$ in $\mathcal{M'}$. It is also clear that the weight of the semi-matching thus constructed is the same as the valuation of the allocation, because for any $i$, the valuation of the allocation for bidder $i$ equals the weight of the edge terminating at vertex $i$ included in $\mathcal{M'}$. Therefore, by Lemma~\ref{lem:auction-approx} weight of the constructed semi-matching $\mathcal{M'}$ is at least $\frac{1}{2} \mbox{OPT}_\mathcal{I}$.
\end{proof}

The remaining part in the proof of Theorem~\ref{thm:edge-weighted} is to show how can we go from a semi-matching of weight $w$ to a matching of weight at least $w/2$ online.
Consider Algorithm~\ref{alg:edge-matching} that takes a semi-matching $\mathcal{M'}$ and outputs a matching $\mathcal{M}$ in an edge-weighted $\mbox{OWNBM}$.

\begin{algorithm}[t]
	\SetAlgoNoLine
	\caption{Constructing matching from semi-matching for edge-weighted $\mbox{OWNBM}$}
	\label{alg:edge-matching}
	\KwIn{A semi-matching $\mathcal{M'}$ in an edge-weighted $\mbox{OWNBM}$ $\mathcal{I}$}
	\KwOut{A (random) matching $\mathcal{M}$ in $\mathcal{I}$ with expected weight at least half the weight of $\mathcal{M'}$}
	$\mathcal{M} \leftarrow \emptyset$\;
	\For{each vertex $i$ in increasing order}
	{
		\If{$(j, i) \in \mathcal{M'}$ for some vertex $j$}
		{
			\eIf{$(i, k) \in \mathcal{M'}$ for some vertex $k$}
			{
				\If{color($k$) == green}
				{
					color($i$) = red\;
				}
				\If{color($k$) == red}
				{
					color($i$) = green\;	
				}
			}
			{
				choose color($i$) green or red with probability 1/2 each\;
			}
			\If{color($i$) == green}
			{
				$\mathcal{M} \leftarrow \mathcal{M} \cup (j, i)$\;
			}
		}
	}
	return $\mathcal{M}$;
\end{algorithm}

\begin{lemma}
\label{lem:edge-matching-weight}
Let $\mathcal{I} = (\mathcal{G}, d)$ be an instance of edge-weighted $\mbox{OWNBM}$ and $\mathcal{M'}$ be a semi-matching in $\mathcal{I}$ of weight $w$. Then Algorithm~\ref{alg:edge-matching} constructs a random matching $\mathcal{M}$ in $\mathcal{I}$ of expected weight $w/2$.
\end{lemma}

\begin{proof}
We need to show that
\begin{enumerate}
\item the constructed set $\mathcal{M}$ is a matching, and
\item weight of $\mathcal{M}$ is $1/2$ of the weight of $\mathcal{M'}$.
\end{enumerate}
Notice that we include edge $(j, i)$ in $\mathcal{M}$ iff color($i$) is green. For (1), assume contrary. Then there exist vertices $i$, $j$ and $k$ such that both the edges $(i, k)$ and $(j, i)$ are included in $\mathcal{M}$. Since we process vertices in order and vertex $k$ must have been processed first. Since edge $(i, k)$ was included in $\mathcal{M}$, color($k$) must be green. This implies color($i$) is red, which contradicts that edge $(j, i)$ was also included in $\mathcal{M}$.

For (2), it can be easily shown by induction that for any vertex $i$, such that there exists some edge $(j, i)$ in $\mathcal{M'}$, color($i$) = green with probability $1/2$. This means the edge $(j, i)$ gets included in $\mathcal{M}$ with probability $1/2$, which proves that the weight of $\mathcal{M}$ is $1/2$ of the weight of $\mathcal{M'}$.
\end{proof}

\begin{proof}[Proof of Theorem~\ref{thm:edge-weighted}]
Combining Lemma~\ref{lem:edge-semi-matching} and Lemma~\ref{lem:edge-matching-weight}, we immediately see that given an instance $\mathcal{I} = (\mathcal{G}, d)$ of edge-weighted $\mbox{OWNBM}$, we can construct a matching $\mathcal{M}$ in $\mathcal{I}$ of weight $\frac{1}{4} \mbox{OPT}_\mathcal{I}$. Since the semi-matching constructed in Lemma~\ref{lem:edge-semi-matching} picks an edge $(j, i)$ no later than $d$ time steps after the arrival of $i$, and since Algorithm~\ref{alg:edge-matching} processes edges in increasing order of the terminal vertex, this yields a $(\frac{1}{4}, d)$-competitive algorithm for edge-weighted $\mbox{OWNBM}$.
\end{proof}

\section{Vertex-Weighted $\mbox{OWNBM}$}
\label{sec:vertex-weighted}
In this section, we prove Theorem~\ref{thm:vertex-weighted}. 

\begin{algorithm}[t]
	\SetAlgoNoLine
	\KwIn{An instance $\mathcal{I} = (\mathcal{G}, d)$ of vertex-weighted $\mbox{OWNBM}$}
	\KwOut{A semi-matching $\mathcal{M'}$ in $\mathcal{I}$}
	let $n$ be the number of vertices in $\mathcal{G}$\;
	define function $g(Y) = e^{Y - 1}$\;
	$\mathcal{M'} \leftarrow \emptyset$\;
	assume all vertices in $\mathcal{G}$ are colored white\;
	pick $type$ as $origin$ or $destination$ with probability $1/2$ each\;
	\eIf{$type == destination$}
	{
		\For{vertex $j$ in increasing order}
		{
			pick $Y_j \in [0, 1]$ uniformly at random\;
			let $N(j)$ be the set of vertices $k$ of color white such that $(j, k)$ is an edge of $\mathcal{G}$\;
			\If{$N{j} \neq \emptyset$}
			{
				let $i = \argmax\{w_k (1 - g(Y_k)): k \in N(j)\}$\;
				$\mathcal{M'} \leftarrow \mathcal{M'} \cup \{(j, i)\}$\;
				color vertex $i$ black\;
			}
		}
	}
	{
		add $d$ dummy vertices of weight $0$ with no edges last in the ordering\;
		\For{vertex $t$ in increasing order}
		{
			pick $Y_t \in [0, 1]$ uniformly at random\;
			let $i = t-d$\;
			\If{$i > 0$}
			{
				let $N(i)$ be the set of vertices $k$ of color white such that $(k, i)$ is an edge of $\mathcal{G}$\;
				\If{$N{i} \neq \emptyset$}
				{
					let $j = \argmax\{w_k (1 - g(Y_k)): k \in N(i)\}$\;
					$\mathcal{M'} \leftarrow \mathcal{M'} \cup \{(j, i)\}$\;
					color vertex $j$ black\;
				}
			}
		}
	}
	return $\mathcal{M'}$;
	\caption{Constructing semi-matching for vertex-weighted $\mbox{OWNBM}$}
	\label{alg:vertex-semi-matching}
\end{algorithm}

\begin{lemma}
\label{lem:vertex-semi-matching-online}
Let $\mathcal{I} = (\mathcal{G}, d)$ be an instance of vertex-weighted $\mbox{OWNBM}$. Then Algorithm~\ref{alg:vertex-semi-matching} constructs a (random) semi-matching $\mathcal{M'}$ in $\mathcal{I}$ online such that any edge $(j, i) \in \mathcal{M'}$ terminating at vertex $i$ is picked no later than $d$ time steps after its arrival. 
\end{lemma}

\begin{proof}
Obvious from the algorithm.
\end{proof}

We could lower bound the weight of the semi-matchings constructed by Algorithm~\ref{alg:vertex-semi-matching}. But for our proof, we need a stronger claim for which we define the notion of {\em half-weight}.

\begin{definition}
\label{def:half-weight}
A semi-matching $\mathcal{M'}$ in an instance $\mathcal{I}$ of vertex-weighted $\mbox{OWNBM}$ constructed by Algorithm~\ref{alg:vertex-semi-matching} is said to be of {\em origin-type} if the random variable $type$ took the value $origin$; else it is said to be of {\em destination-type}. The half-weight of an origin-type semi-matching is defined as 
\[
	\sum_{(j, i) \in \mathcal{M'}} w_j;
\] 
while the half-weight of a destination-type semi-matching is defined as
\[
	\sum_{(j, i) \in \mathcal{M'}} w_i.		
\]
\end{definition}

We use a beautiful primal-dual proof technique introduced by Devanur, Jain and Kleinberg~\cite{DJK13} for lower bounding the expected half-weight of a semi-matching constructed by the algorithm. However, the proof is omitted in this shorter exposition.

\begin{lemma}
\label{lem:vertex-semi-matching-half-weight}
Let $\mathcal{I} = (\mathcal{G}, d)$ be an instance of vertex-weighted $\mbox{OWNBM}$ and $\mbox{OPT}_\mathcal{I}$ be the weight of an optimal matching. Then Algorithm~\ref{alg:vertex-semi-matching} constructs a (random) semi-matching $\mathcal{M'}$ of expected half-weight $\frac{1}{2} (1 - \frac{1}{e}) \mbox{OPT}_\mathcal{I}$.
\end{lemma}

To complete the proof of Theorem~\ref{thm:vertex-weighted}, we need to show how can we go from a semi-matching of half-weight $w$ to a $3$-matching of weight at least $w$ online. The proof of the following lemma is omitted.

\begin{lemma}
\label{lem:vertex-matching-weight}
Let $\mathcal{I} = (\mathcal{G}, d)$ be an instance of vertex-weighted $\mbox{OWNBM}$ and $\mathcal{M'}$ be a semi-matching in $\mathcal{I}$ of half-weight $w$. Then we can construct a $3$-matching $\mathcal{M}$ in $\mathcal{I}$ of weight at least $w$ by processing vertices in order. An edge $(j, i)$ is picked in a $2$-set or a $3$-set while processing vertex $i$, and may get deleted from a $3$-set while processing vetex $j$.
\end{lemma}

\begin{proof}[Proof of Theorem~\ref{thm:vertex-weighted}]
Combining Lemma~\ref{lem:vertex-semi-matching-half-weight} and Lemma~\ref{lem:edge-matching-weight}, we immediately see that given an instance $\mathcal{I} = (\mathcal{G}, d)$ of vertex-weighted $\mbox{OWNBM}$, we can construct a $3$-matching $\mathcal{M}$ in $\mathcal{I}$ of weight $\frac{1}{2} (1 - \frac{1}{e}) \mbox{OPT}_\mathcal{I}$. Since the semi-matching constructed in Algorithm~\ref{alg:vertex-semi-matching} picks an edge $(j, i)$ no later than $d$ time steps after the arrival of $i$, and since Lemma~\ref{lem:vertex-matching-weight} processes edges in increasing order of the terminal vertex, this yields a randomized $(\frac{1}{2} (1 - \frac{1}{e}), d)$-competitive algorithm for vertex-weighted $\mbox{OWNBM}$ where an edge $(j, i)$ picked in a $3$-set may get deleted no later than $d$ steps since the arrival of vertex $j$.
\end{proof}

\section{Conclusions}
\label{sec:conclude}
In this work, we initiated the study of matching rides in a ride-sharing platform in a formal graph-theoretic online matching setting. For this purpose, we proposed a model that we call Online Windowed Non-Bipartite Matching ($\mbox{OWNBM}$). Our model is simple and elegant, and yet remarkably powerful to capture the important attributes of the rider matching problem in ride-sharing platforms. We showed how the edge-weighted and vertex-weighted versions of $\mbox{OWNBM}$ capture the objectives of matching efficiency and matching rate, which are important and widely used metrics in ride-sharing industry. 

We provided a randomized $\frac{1}{4}$-competitive algorithm for the edge-weighted version, and a randomized $\frac{1}{2} (1 - \frac{1}{e})$-competitive algorithm for the vertex-weighted version of $\mbox{OWNBM}$ (with some constraints relaxed). While we believe our bounds might be further improved, these are first results for online matching in non-bipartite graphs in the adversarial setting to the best of our knowledge. Our result for the edge-weighted version uses a beautiful result from combinatorial auction theory, while that for the vertex-weighted version uses insights and techniques from a very elegant randomized primal-dual analysis technique in matching theory.

We believe our model is of independent theoretical interest and can find use outside of the immediate context of ride-sharing. In fact, one of the novelties of our work is proposing a model for studying online matching in non-bipartite graphs. We leave a range of open questions and future research directions. We believe it is possible to improve the bounds we presented. Our model can also be studied in the stochastic setting where edges are drawn from a distribution which is known or can be learned. We considered two-party matches in this work. One can consider hyper-edges in matching which correspond to matching more than two parties together for sharing a car. We believe that these and other related challenging problems are both theoretically interesting and practically relevant.

\bibliographystyle{plainurl}
\bibliography{online_matching}

\end{document}